\documentclass[lettersize,journal]{IEEEtran}
\usepackage{amsmath,amsfonts}
\usepackage{algorithmic}
\usepackage[ruled,linesnumbered,vlined]{algorithm2e}
\usepackage{float}
\usepackage{comment}
\usepackage{array}
\usepackage[caption=false,font=normalsize,labelfont=sf,textfont=sf]{subfig}
\usepackage{textcomp}
\usepackage{stfloats}
\usepackage{url}
\usepackage{verbatim}

\usepackage{graphicx}
\usepackage{cite}
\usepackage{tabularx,booktabs}
\usepackage{amsmath,amssymb,amsfonts}
\usepackage{algorithmic}
\usepackage{graphicx}
\usepackage{textcomp}
\usepackage{comment}
\usepackage{orcidlink}
\usepackage{forest}
\usepackage{soul}
\usepackage{amsthm}
\usepackage{pgf-umlsd}
\usepackage{tikz}
\usetikzlibrary{
	graphs,
}
\definecolor{bbord}{RGB}{70,78,82}     
\definecolor{ffill}{RGB}{171,213,238}  
\definecolor{eedııge}{RGB}{27,128,196}   
\definecolor{medge}{RGB}{96,15,36}     
\tikzset{st1/.style={blue,dashed, draw=blue, fill=yellow!40, minimum size=0.5cm}}
\tikzset{st2/.style={red,dashed, draw=red, fill=yellow!40, minimum size=0.5cm}}
\tikzset{st3/.style={black,dashed, draw=black, fill=yellow!40, minimum size=0.5cm}}

\newtheorem{theorem}{Theorem}[section]

\theoremstyle{proposition}
\newtheorem{proposition}[theorem]{Proposition}
\theoremstyle{corollary}
\newtheorem{corollary}[theorem]{Corollary}
\theoremstyle{definition}

\newtheorem{example}[theorem]{Example}

\usepackage{etoolbox}

\theoremstyle{remark}
\newtheorem{remark}[theorem]{Remark}

\usepackage[utf8]{inputenc}
\usepackage[T1]{fontenc}
\usepackage[english]{babel}
\usepackage{multicol}
\usepackage{tabularx}
\usepackage{multirow}
\newcolumntype{Y}{>{\centering\arraybackslash}X}

\ifCLASSINFOpdf

\else

\fi

\numberwithin{equation}{section}

\begin{document}

\title{ Group Authentication and \\Key Establishment Scheme}

\author{Sueda~Guzey,~\IEEEmembership{Graduate Student Member,~IEEE,}
	Gunes~Karabulut Kurt,~\IEEEmembership{Senior Member,~IEEE,}
	and~Enver~Ozdemir,~\IEEEmembership{Senior Member,~IEEE}
	\thanks{Sueda Guzey and Enver Ozdemir are with Informatics Institute, Istanbul Technical University, Istanbul, Turkey, e-mails:\{kayci19, ozdemiren\}@itu.edu.tr.}
	\thanks{G. Karabulut-Kurt is with  the Department of Electrical Engineering,  Polytechnique Montr\'eal, Montr\'eal, Canada, e-mail: gunes.kurt@polymtl.ca.}
	\thanks{}}




\maketitle

\begin{abstract}
	 Group authentication is a technique that verifies the group membership of multiple users and establishes a shared secret key among them. Unlike the conventional authentication schemes that rely on a central authority to authenticate each user individually, group authentication can perform the authentication process simultaneously for all the members who participate. Group authentication has been found to be a suitable candidate for various applications in crowded in Internet of Things (IoT) environments, such as swarms of drones for agriculture, military, surveillance, where a group of devices need to establish a secure authenticated communication channel among themselves.	The recently presented group authentication algorithms mainly exploit Lagrange polynomial interpolation along with elliptic curve groups over finite fields.  A polynomial interpolation-based group authentication scheme has a vulnerability that allows malicious interruption by any single entity in the process. Moreover, this scheme requires each entity to obtain the tokens of all other entities, which is impractical in a large-scale setting. The cost of authentication and key establishment also depends on the number of users, creating a scalability issue. As a fresh approach to eliminate these issues, this work suggests the use of  inner product spaces for group authentication and key establishment. The approach with linear spaces introduces a reduced computation and communication load to establish a common shared key among the group members.   In addition to providing lightweight authentication and key agreement, this  approach allows any user in a group to make a non-member a member, which is expected to be useful for autonomous systems in the future. The scheme  is designed in a way that the sponsors of such members  can  easily be recognized by anyone in the group. Unlike the other group authentication schemes based on Lagrange's polynomial interpolation, the proposed scheme doesn't provide a tool for adversaries  to compromise the whole group's secrets by using only a few members' shares as well as it allows to recognize a non-member easily, which prevents  the denial of service  attacks from which the former group authentication algorithms suffer. 

\end{abstract}


\section{Introduction}
\IEEEPARstart Internet of things (IoT) networks are composed of lightweight sensors and actuators that can exchange sensing data for communication and computation purposes. From this aspect, they can be considered as a distributed cyber-physical system, and in such networks, the reliability of the sensing data is critical for the safety and reliability of the network. Hence, all of the network nodes should be authenticated so that only the authorized nodes can transmit and receive data \cite{demir2020garden}. The authentication is the process  of verifying the identity of an entity; in other words, it is  a process for deciding  whether  or not a  user is really  who it claims to be \cite{SSL,PGP}. Since it has a vital  role in  establishing secure communication and confidential data transmission,  an efficient, reliable, and scalable  authentication process is a necessity for IoT networks.    
In the presence of a densely populated network with a high number of devices connected to each other, as in an IoT network such as  massive machine type communication (mMTC) \cite{9993704}, the authentication task becomes a challenging process in terms of  the need for the storage, energy, communication, and computation power, that is yet to be solved. Aligning with the authentication process, the additional security measure, called  the secret key establishment among  users in a network, puts an unbearable computational  burden especially on  the resource constraint devices in case of employing standard cryptographic primitives. In order to reduce this high computational complexity of authentication and key agreement processes, group authentication schemes (GASs) have been presented as a solution instead of one-to-one authentication~\cite{Harn}. { A group authentication scheme is a potential candidate for environments that consist of devices that need to establish a secure channel among themselves, such as swarms of drones in dynamically changing environments.\\

\indent GAS is a process that has a similar purpose  as authentication algorithms; namely, a GAS is a process of confirming whether several users belong to a predetermined group. The first group authentication scheme was inspired by a threshold secret sharing scheme and used polynomial interpolation \cite{Harn}. It presented three different algorithms, each of which identified a user from its public information. The user’s private key was the image of a secret function chosen by the group manager. In other words, each user has a key pair $(x,f(x))$, where the first component is the public key and the second component is the private key. The constant term of $f(x)$ could be easily obtained by combining the public and private information of $\deg(f(x))+1$ or more users. When the users computation results in the constant term, the authentication phase is completed. The first group authentication algorithm updates the keys of users after each authentication process, but it is vulnerable to malicious behavior. In other words, an algorithm based on polynomial interpolation cannot identify the malicious party or parties if the multi-party computation fails to return the constant term of $f(x)$. The successor methods eliminated the need to update keys after each authentication process, but the problem of failing to complete the authentication process persisted when the process resulted in a value different from the predetermined constant term of the function. Therefore, the group authentication schemes based on polynomial interpolation are susceptible to denial of service attacks (DOS). In addition to this, they have a high communication load, as each member had to share keys with every other member who joined the authentication process.
	
Group authentication schemes  are a suitable alternative to sequential authentication methods, which require high energy consumption and computational power. GASs allow multiple users to authenticate themselves as a group with a single message. Even though many variations of GAS have been developed \cite{Harn,YA},  most of the existing GASs rely on Lagrange’s polynomial interpolation, which poses several challenges, such as: 
\begin{enumerate}[]
	\item Scalability: The algorithm becomes inefficient when the group size is large, as it requires more computations and communications.
	\item Communication overhead: The users need to exchange secret keys with each other, which increases the network traffic and the risk of leakage.
	\item Computational complexity: The security of the GASs depends on the hardness of the discrete logarithm problem on elliptic curves, which requires more arithmetic operations than on finite fields.
\end{enumerate}

In the following part, we  discuss these problems in more detail and propose a novel GAS that overcomes them.

\textit{Scalability}: Group authentication is a practical application that may involve a large number of users in a group. As more devices are connected to each other, group authentication and key-sharing methods may face difficulties. The existing algorithms have to balance between security and operational cost, and it is hard to find the optimal solution for each case. The main problem is that an adversary can compromise the group security by obtaining some members’ private information, which can be combined to reveal the secrets of all users in the group.  The current algorithms based on polynomial interpolation have a linear relationship between the cost of authentication and the number of users in the process. In this paper, we propose a novel algorithm that achieves scalability as the cost is independent of the number of users in the process and high security.\\
\indent \textit{Communication overhead}: In  GASs, which are inspired by Shamir's secret sharing algorithm \cite{SSS}, the public keys are random numbers, and private keys are  the images of these random numbers under the group manager's private function.  The main drawback of all these approaches is to  require several communication to construct the group secret. The group secret has to be obtained by each member of the group in order to communicate in a secure way among the members of the group. Each member, therefore, goes through the same process to have this secret key. This process  includes sharing secrets with  all other members. In addition to high communication costs, unfortunately, such a sharing operation  might give an opportunity to an eavesdropper to capture the group's secret.\\
\indent \textit{Computational complexity}:  In order to add an extra layer of security, recent methods utilize elliptic curves over  finite fields along with the polynomial interpolation \cite{YA}. The security of the system then would depend on the hardness assumption of discrete logarithm problem (DLP) in an elliptic curve group over a finite field. On the other hand, interpolating points in the Euclidean plane itself might be considered to be costly, let alone involving group operation in an elliptic curve group. A single addition or doubling in an elliptic curve group costs more than 12 multiplications in a finite field.   Even though these operations are acceptable in terms of energy usage and computational power at the beginning, frequent authentication or key establishment  might not be bearable in some situations where power constraint IoT devices involve.

The following issues in the first and the second generation group authentication schemes form a motivation to open a new avenue toward research in this direction:
\begin{enumerate}
	\item The earlier versions of the group authentication algorithms require the polynomial interpolation to extract the secret, which means each user needs to obtain a certain number of others share, which puts a communication burden on each user for authentication and secret key establishment.
	\item The requirement of other members' shares makes algorithms based on polynomial interpolation vulnerable to the denial of service (DOS) attack. In fact, any of the current group authentication schemes can not determine an intruder joining the authentication process.
	\item The whole group's security depends on the manager's secret polynomial which can be revealed by a certain number of members. In other words, the schemes based on polynomial interpolation are vulnerable to the well-known Sybil attacks.
	\item The cost of the key establishment and the authentication process also depends on the manager's secret polynomial. For example, if the degree is large, each member needs to accumulate a large number of members' shares which increases communication and computational load. In case a small degree polynomial is in charge, then the chance to exploit Sybil attacks  increases dramatically. 
	\item The final problem involves a group manager for the registration phase in the algorithms based on the polynomial interpolation. This obligation prevents the group authentication schemes from exploiting especially decentralized and autonomous systems.  
\end{enumerate}

In order to jointly address the aforementioned  motivations, the proposed  GAS offers a new mathematical tool for authentication and key establishment. The algorithm is based on the well-known inner product spaces and projection operation. The scheme offers a group authentication algorithm that is scalable, lightweight, and secure. It also minimizes the operations needed for key establishment and authentication in a group, regardless of its size. We implemented all algorithms using the same C++ library, PARI/GP\cite{PARI}, for high-precision operations. Our results show that the proposed approach has acceptable computational costs even with 10000 users. Thereby, it is suitable to be adapted by IoT devices with limited resources. In addition to these advancements, in this proposal, the private information of each device is independent of one another. In other words, even if an adversary obtains all members' private keys but one, the adversary can not extract any information about  that member's secret key.

\subsection{Advantages of the Proposed  Scheme:}
The proposed lightweight group key establishment scheme solely relies on inner product operations, which might require univariate polynomial arithmetic depending on the selected inner product space. Due to the nature of  inner product spaces, the proposed algorithm encompasses the following advantages that are not offered by well-known group authentication algorithms. These  advantages make the proposed scheme a likely candidate for   practical authentication and key establishment in communication systems.  
\paragraph*{Advantage 1} The group key establishment doesn't require having other members share for each individual in the group.  In other words, publicly known information released by any member of the group is enough to extract the group secret. In this way, the security risk coming from exchanging members' shares among peers is removed completely.  This eliminates the security risk and the communication overhead of exchanging shares among peers.
\paragraph*{Advantage 2} The secure communication among the members of a group first requires the authentication of members if a usual group authentication algorithm is employed. The key establishment of the proposed scheme is set up  in a way that a non-member can not continue exchanging data with the members.  Since an infiltrator or a  non-member can not extract the key, it eliminates the need for additional group membership confirmation.
\paragraph*{Advantage 3} The proposed scheme has a constant cost for key extraction and authentication, regardless of the group size. However, other group authentication algorithms that use polynomial interpolation depend on the number of users for authentication and constructing the group secret key. This can lead to costly operations in some cases \cite{Harn}. Even the scalable GAS proposed by Aydin et al. \cite{YA} still requires the function to be constructed with the users who participate in the authentication process.
\paragraph*{Advantage 4} The proposed method allows any member to add a non-member to the group, and the sponsor of the new peer is easily identifiable by the group. In contrast, other group authentication algorithms require a member to have the same privileges as the group manager to add a new member. This means that the member must know the function that the group manager initially chose. Anyone with this function can add any user to the group, but the sponsor of the new member remains unknown. The group authentication schemes that use polynomial interpolation also need a central authority or its equivalent to register a user.
\paragraph*{Advantage 5}  The security of group communication in a GAS that uses polynomial interpolation depends on the group secret function, which the group manager generates. However, this function can be exposed by combining the private information of some users. Thus, an adversary who has access to some users’ data can breach the security of the whole group. This also compromises the private keys of all members. In contrast, the proposed method protects the individual secrets of each member, even if an adversary obtains some private keys. Therefore, the adversary cannot access the other members’ private information.
\paragraph*{Advantage 6} The GASs based on the idea of secret sharing authenticates users by combining certain members' shares in case the group manager is not available. Only when all shares' are legitimate then the method confirms the users. In other words, even if one user is not legitimate, existing GASs cannot continue the authentication process, and the method can not pinpoint the illegitimate users in the group. {This  causes an interruption of the service even if a single adversary attacks to the group authentication process. The proposed method allows any member of the group to locate a non-member easily, and this prevents the group from  a DOS attack \cite{ddos}.  

The remaining part of the paper is organized as follows.  In Section II, related studies in the general area of authentication and, in particular, group authentication and mathematical background are presented. Section III is spared for describing the proposed method.   The security analysis is presented in Section IV. Running time comparison is given in Section V. Conclusion, and future plan are given in Section VI.     
\section{Background}

\indent The increase in the number of  communication nodes in the digital medium compels researchers to search for non-traditional methods in secure  communication \cite{IoTSoc}. The first constraint that stands out during the designing of secure protocols is the presence of devices with limited energy and computational resources.  Therefore, traditional methods employing a public key algorithm \cite{DH, ElGamal, RSA} for authentication and key agreement are no longer suitable due to several reasons: The first one is the requirement of high computational load during the implementation of public key algorithms and the second one is the requirement of  responding to all device’s requests separately. In addition, utilizing any of these public key algorithms in an authentication scheme also requires the presence of a  certification authority. Indeed, the certification is indispensable to prevent the well-known Man-in-the-Middle Attack \cite{maninthemiddle}. In addition to the high computational power need, the communication load with the certification authority might bring an unbearable   burden. Considering all these problems with the usual authentication schemes,  GAS might be the most suitable solution for the resource constraint devices.

One of the most useful tools for the purpose  of constructing a GAS is  a threshold scheme based on polynomial interpolation. The first study on the threshold secret sharing scheme was presented by  Shamir \cite{SSS}, where a secret is divided into a number of pairs to be distributed among the shareholders. 
The secret can only be recovered by anyone holding as many shares as the threshold value. An adversary can not obtain any information about the secret unless it has more shares than the threshold value.\\
\indent The work proposed by Harn \cite{Harn} exploits  the polynomial interpolation in the three different authentication schemes as in the case of the secret sharing algorithm.    In his first scheme, which is called a synchronous $(t,m,n)$-group authentication scheme,  a polynomial $f(x)$  of degree $t-1$ is selected, and the constant term of it, say  $s$, is set to be the secret group key by the group manager $(GM)$.  Then the $GM$  calculates the private key $f(x_i)$ and conveys it with the public key $x_i$ to  the corresponding user $U_i$ for each $i=1,\dots,n$. The number of users joining the authentication process is $m \ge t$ out of $n\ge m$ users.  To be authenticated,  each user releases its tokens, and once each user has other users' shares, they can compute the group secret via:
$$ s= \sum_{i=1}^m f(x_i) \prod_{r=1 r\not=i} ^{m} \frac{x_r}{x_i - x_r}.$$
This method is secure when private key sharing is done simultaneously; otherwise, any adversary can obtain the group key by having $t$ or more shares via constructing the polynomial with these shares. Harn has proposed the second scheme in case sharing is asynchronous. In the token generation phase, the  $GM$ selects $k$ random polynomials $f_{l}(x)$   having   degree  $t-1$ such that $kt>n-1$  for $l=1,2,\dots,k$.  Then, it sets the secret key for user $U_i$ as $f_{l}(x_i)$ for  $l=1,2,\dots,k$. The group manager selects random $w_j$ and $d_j$ such that  $$ s= \sum_{j=1}^{k} d_{j} f_{j}(w_j). $$ The group manager then broadcasts $w_j,d_j $ for    $j=1,2,\dots,k$.
For authentication,  each user $U_i$ computes\\
$$ c_i= \sum_{j=1}^{k} d_j f_{j}(x_j) \prod_{r=1 r\not=i} ^{m} \frac{w_j - x_r}{x_i - x_r}.  $$ 
then releases $c_i$.   Each user joining the authentication process adds up the released information as
$$ s'=  \sum_{r=1}^{t} c_r. $$ In the following step, the users verify whether the equality, $H(s')= H(s)$, holds, and then the authentication process is completed; that is,  all users have been authenticated. Note that all these operations take place in a finite field, and an adversary can not obtain any information about private tokens of users from $c_i$ or $w_i$.  However, this method only allows members to use their tokens only once. Harn proposed another authentication scheme that allows members to use their tokens multiple times \cite{Harn}.
In this third method, the group manager first selects two large prime $p$ and $q$ such that $q$ divides $p-1$, a generator $g$ of $GF(q)$ and two polynomials $f_l(x)$ for $l=1,2$ which  have degree $t-1$. $GM$ then selects random integers  $w_{i,j}$ , $d_{i,j}$ for $j=1,2$ and sets the secret $$s_i= g_i^{ \sum_{}^{}d_{i,j}f_j(w_{i,j}) }.$$ The randomly chosen integers $w_{i,j},d_{i,j}$ and the hash values $H(s_i)$ are made publicly known by $GM$. For group authentication each user $U_i$  computes  $$c_i=\sum_{}^{}d_{i,j}f_j(x_i)\prod_{r=1\\  ,i\not=j}^{m} \frac{w_{i,j}-x_r}{x_i-x_r}$$ via their tokens and then finds $e_i=g_i^{c_i}$ to share other users in the group. Once the users have all $e_i$ for $i=1,2,\dots,m$  each one computes 
$$ s_i'=\prod_{i=1}^{m} e_i ,$$ and check if $H(s_i)=H(s_i')$. If equality holds, authentication of all users in the process is done. Otherwise, there must be at least one user  who is not a group member. Note that an attacker can not obtain any information  about $c_i's$ by having $e_i's$ thanks to the hardness assumption of discrete logarithm problem in the multiplicative group $GF(q)$.
 
 All the methods summarized above have certain vulnerabilities which prevent them from being employed in practice, and a new method for group authentication has been introduced in \cite{YA}. The elliptic curve discrete logarithm problem (ECDLP) is utilized in this work to provide a certain  security level for the group authentication algorithm.  In the initialization phase, the group manager determines   a cyclic group $G$, a generator $P$ for it, an encryption and a decryption algorithms $E(\cdot)$, $D(\cdot)$, and a hash function $H(\cdot)$.  $GM$ also selects a  polynomial $f(x)$ of degree $t-1$ whose constant term is the master secret $s$. Each user, $U_i$ for $i=1,2,\dots,n$ in the group, has one public information $x_i$ and one private information $f(x_i)$. Lastly, the group manager computes the value $Q=sP$ and makes  $P,Q,E(\cdot),D(\cdot),H(\cdot),H(s)$ and $x_i$'s publicly known and shares $f(x_i)$ with the user $U_i$ privately.  For $GM$ handling authentication, each user  $U_i$ computes $f(x_i)P$ and sends it to $GM$ in the group by concatenating this information with its identification number. This  prevents the public share $f(x_i)P$ from being used by any other user in future communications. If the $GM$ is responsible for the authentication part, it computes $f(x_i)P$ for each user $U_i$ and compares results with the received values. In case all are valid, the verification is completed successfully. Otherwise, the users who are not group members can be determined by $GM$. If $GM$ is not included in the verification phase, any user collects $f(x_i)P$ from others in the group  and can handle authentication by computing $$C_i=\left(\prod_{r=1,r\not=i}^{m} \frac{-x_r}{x_i - x_r}\right)f(x_i)P.$$ This verifier node checks whether $$\sum_{i=1}^{m}C_i \overset{?}=Q .$$ The authentication succeeds if the equation is satisfied. Otherwise, one or more non-members attempt to join the authentication phase, but the intruder(s) cannot be identified. We should note that, as in Harn’s group authentication method, it is impossible to determine which user or users are not in the group.\\
\indent There are several other studies on authentication based on Shamir's secret-sharing algorithms. One of them is a selective group authentication scheme for IoT-based medical information systems\cite{healtcaresystem}. This proposal aims to solve the security problems in healthcare services  such as  misuse of medical devices or illegal access to a medical service. For this purpose, a  group authentication scheme using Shamir's threshold technique is presented. Still, it is not suitable for  resource-limited devices since  a lot of communication is needed, even for a single user to be authenticated. Another secret sharing-based group authentication study is \cite{graycode}. This work uses the Gray code to construct the shares and the XOR operation to reconstruct the secret. This differs from conventional secret-sharing studies, which do not specify how to share the key among group members. Moreover, the proposed key establishment scheme only works for groups of 3 or 7 members. The protocol in \cite{vandermonde} employs a linear secret sharing scheme using the Vandermonde matrix  instead of the classic version  to distribute  pairs of the group key.  The purpose of this work is to reduce the computation load of the group authentication phase for energy-constrained IoT devices.

There is also another mathematical approach to group authentication \cite{CRT}. This proposal  is based on the Chinese remainder theorem (CRT). That is if any user has shadows up to $r$, the secret value $y$ can be computed using CRT. Another group authentication study that uses the Paillier threshold cryptography as  a tool  is proposed in \cite{mahallethreshold}. They have compared the running time of their work with Harn's algorithm and showed that their experimental results are better than Harn's work. However, there is a point to note here that they did not count the cost of public and private key encryptions, and the scalability issue is not considered. Apart from mathematical-based algorithms, various other types of algorithms have been proposed for authentication. A machine learning tool along with biometrics has been proposed to perform authentication in IoT networks \cite{ML}. The method asks users to have a certain share to be authenticated, and it is only suitable for small-size groups. An authentication method for dynamic groups has been investigated in \cite{dynamic}. The method requires aggregation of users' shares to perform authentication in machine-type communication. A lightweight authentication method is presented especially for machine-to-machine communication in \cite{24}. Each user performs computations to obtain their authentication codes, which convey to the group manager, and the group manager authenticates users based on the received codes. 

 A recent group authentication scheme is presented in  \cite{gase} and it is based on secret sharing scheme which relays on polynomial interpolation. The work seeks to enhance the security of EC systems by shifting data processing closer to the data source, departing from the conventional centralized approach.  Involving several edge servers increases the possibility of edge-servers compromise. For this purpose, this work proposes a group authentication scheme based on secret sharing providing a way for refreshing the session keys. However, utilizing secret sharing brings some problems with itself such as being vulnerable to DOS attacks, not being able to detect illegitimate users participating in the authentication process, and causing communication overhead during the sharing phase between users. There is another group-based authentication scheme for machine type communication (MTC) in LTE-A networks \cite{MTC}.  Similar with \cite{gase}, the secret sharing method is utilized in authentication process and the same problems mentioned above might arise.\\ 
\indent A group authentication scheme utilizing historical collaboration process information is presented in \cite{collaboration}. The scheme leverages knowledge acquired during the previous collaboration round to generate tokens for mutual authentication among all devices before the next collaboration round commences. This makes the current authentication session depent on the previous one.
	 The another study \cite{drones} proposes a group authentication protocol for drones systems. This authentication and key agreement protocol aims to provide a secure way for data transmission  between swarms and Zone Service Providers (ZSPs) over an insecure communication channel. The algorithm employed in this study incorporates bitwise XOR, hashing, and PUF operations, rendering it lightweight. However, in scenarios involving ZSPs situated in diverse environments, this protocol may not be suitable for authentication. In simpler terms, this protocol does not facilitate cross-domain group authentication, wherein drones authenticate ZSPs located in disparate physical spaces.
For addressing the issue of cross-domain group authentication, the study \cite{cross_domain} has employed blockchain techniques. Specifically, the authentication scheme proposed in this paper relies on cooperative blockchains (BCs), including intra-domain and inter-domain BCs. The intra-domain BC is responsible for recording legitimate users’ registration and authentication information within a single domain, while the inter-domain BC serves the same purpose for cross-domain authentication. On the other hand, these blockchain-based security protocols necessitate frequent updates to the cryptographic information stored within the blockchain. This results in a substantial communication and computation overhead. 

The algorithm in this work exploits the inner product on a vector space. For the sake of completeness, a brief summary of the inner product and orthogonal projection is presented in the following subsections. 

\subsection{Inner Product Space}
Let $E$ be a vector space over a field $\mathbb F$. The distance between two vectors in $E$ is measured through a function called {\it inner product.} An inner product function, $\langle, \rangle$ maps two vectors $v,w$ to an element $\mathfrak f \in \mathbb F$:
  $$ \langle \cdot, \cdot \rangle =E\times E \rightarrow \mathbb F$$
An inner product function must satisfy:
\begin{itemize}
	\item {\it Linearity:} $\langle \alpha v_1+\beta v_2, v_3\rangle= \alpha \langle v_1,v_3\rangle +\beta  \langle v_2,v_3\rangle$ for all $v_1,v_2,v_3\in E$ and $\alpha, \beta \in \mathbb F$. 
	\item {\it Symmetry:} $\langle v_1,v_2\rangle=\langle v_2,v_1\rangle$ for all $v_1,v_2\in E$.
	\item {Positive Definiteness:} $\langle v,v \rangle>0$ if $v\ne 0 \in E$.
\end{itemize}

\subsubsection{Orthogonal Projection}	 
\indent  The approximation problem has been at the center of interest for applied sciences \cite{eigenfaces}, \cite{orthogonalprojection} for centuries. Let $\mathbb E$ be as above and assume for a moment that $\mathbb E$ is  the space of all continuous functions over the real numbers.   For any element $h$ in $\mathbb E$, the important problem is  to find a polynomial or a trigonometric function $g$ which is the closest to $h$. To give an answer for this, one first need to define what the closest means or define a distance function.  A distance function might be constructed easily for the inner product spaces.  Upon defining a proper way to measure distance, the closest function can be constructed by exploiting the inner product  defined on $\mathbb E$ \cite[Section 6.9]{Kincaid}. The method is based on the fact that for a given subspace $P$ of $E$, $g\in P$ is} the best approximation to $h$ if and only of $h-g$ is perpendicular to all vectors in $\mathbb P$, as depicted in Fig. \ref{fig1}.

\begin{figure}[tb] \label{vector projection}

	\centering
	\includegraphics[width=0.5\textwidth]{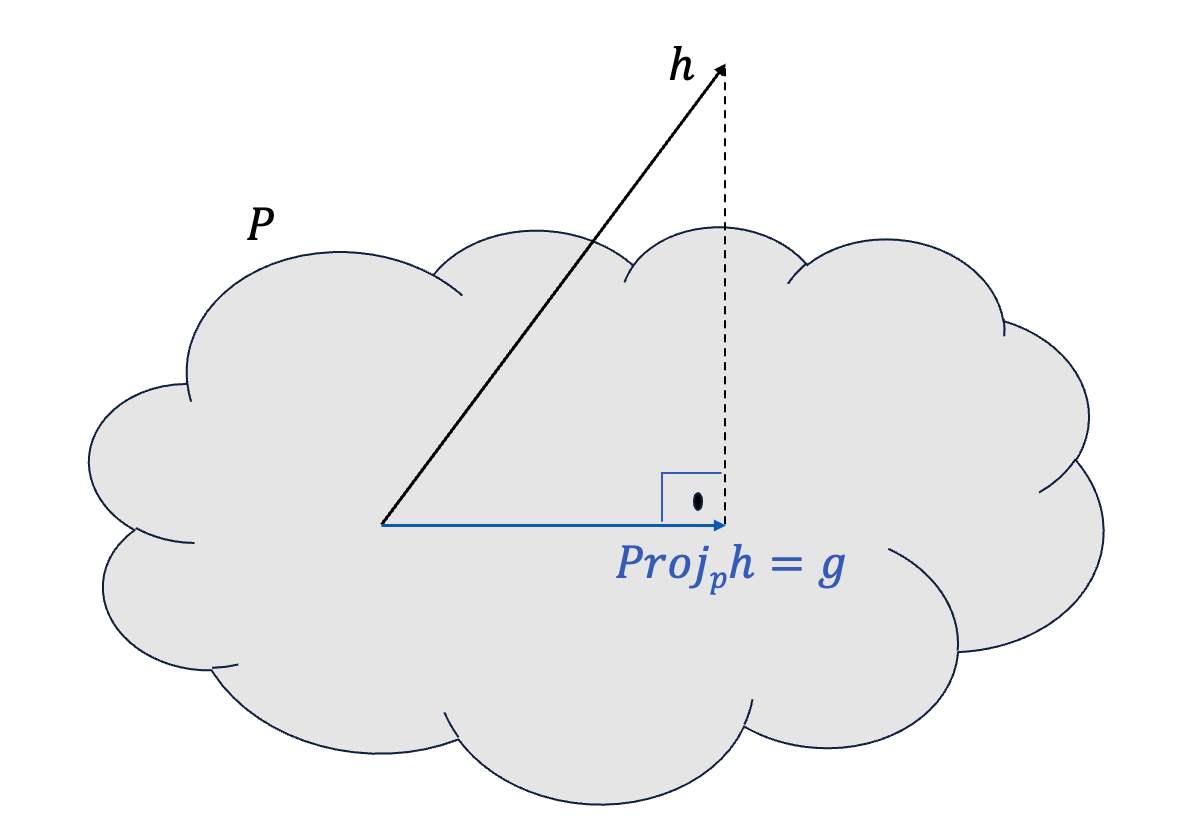}
	\caption{ The best approximation to $h$ is the projection vector.} \label{fig1} 
\end{figure}

In order to find $g\in P$ such that $g-h$ is orthogonal to $P$, one first needs a basis $S=\{b_1,\dots, b_n\}$ for $P$. Then locating the best approximation $g$ is just a computational task.  In fact $g$ can be written as
$$g=a_1b_1+\dots+a_nb_n$$ where $a_i$'s are in the base field for $i=1,\dots, n$. Solving the following system of linear equations for $a_1,\dots, a_n$ gives $g$:
$$ 
\begin{array}{ccc}
	\langle h-g,b_1\rangle&=&0\\
	\langle h-g,b_2\rangle&=&0\\
	\cdot &\cdot &\cdot\\
	\cdot &\cdot &\cdot\\
	\cdot &\cdot &\cdot\\
	\langle h-g,b_n\rangle&=&0.\\
\end{array}
$$      
The other way, or sometimes the easiest way, is to convert the basis $S$ to an orthonormal basis  $S'=\{b'_1,\dots,b'_n\}$  via Gram-Schmidt orthogonalization process. Then
$$c_i=\langle h,b'_i\rangle $$
where   $g=c_1b'_1+\dots+c_nb'_n.$
The computational load in this way occurs while converting $S$ to $S'$ via Gram-Schmidt orthogonalization method \cite{GS}. The cost of operations also depend on the selected inner product space and the inner product itself.

\begin{example}
	A suitable universal set is   a function space $\mathbb E$. An inner product on $\mathbb E$ can be chosen to be the standard one as:
	$$\langle h(x),g(x) \rangle =\int_0^1 h(x)g(x) dx\text{ for all } h(x),g(x)\in \mathbb E. $$ 
	A subspace $\mathbb W$ of $\mathbb E$ can be selected as the set of all polynomials of degree $\le 4$. The proposed algorithm distributes random elements in the selected subspace to the members of the group. For $\mathbb W$, the algorithm selects 5 random number $a_0,a_1,a_2,a_3,a_4$ in the base field $\mathbb F$ and set $$p(x)=a_4x^4+a_3x^3+a_2x^2+a_1x+a_0$$ as the random element in $\mathbb W$.
\end{example}

\section{Proposed Method}
\indent In this work, the inner product space is exploited for a group authentication and key establishment scheme. An  inner product space $E$ is the main object in our scheme. The idea emerged from a realization that a finite-dimensional subspace of a vector space $E$ has infinitely many bases and once a user has any basis for the subspace, it will be privileged as the others having a basis for the same subspace. The secrets are constructed with the predefined subspace, and  apart from group members, no one else can construct the group secret. Once the initial distribution of basis to group members is completed, the group members can construct the secure key and establish a confidential channel. Moreover, the members can privately exchange data with the group manager and another peer in the group.\\

Here we provide an exemplary scenario:
Although the group authentication scheme can be applied to various environments, the current specific use case involves swarms of drones. For instance, a group of drones can be used for surveilling gas pipelines or borders. Hundreds or thousands of such drones need to be operational to obtain live data for the security and availability of the pipelines. These devices have limited battery life and they offer service for a short duration. Therefore, a frequent addition of new drones happens in the group. The newly added members should be able to communicate with each of the drones through a secure channel. In such scenario, a group authentication scheme is one of the practical solutions to apply. Each device's data transmitted should be understood by the others and the communication among the members should be private to the group members. Therefore, a secret key should only be known be the members of the drone groups.
\subsection{System Model}
The system model, which also covers the scenario above, is presented in Figure \ref{fig2}.\\
\indent $GM_i$: The $i^{th}$ group $G_i$'s manager is represented by $GM_i$ as in Figure \ref{fig2}. The group manager might have more processing power than any other member.\\
\indent $U^i_{j}$: Any member in the group $G_i$ is represented by $U^i_j$. A member can be any device equipped with data exchange capability to other members and the group manager.\\
{\it Network Model: }\\
\indent We are assuming that the first registration of a member to a group takes place in a secure environment. Apart from that, any further communication is open to the public.
\begin{figure}[t] 
	
	\centering
	\includegraphics[width=\linewidth]{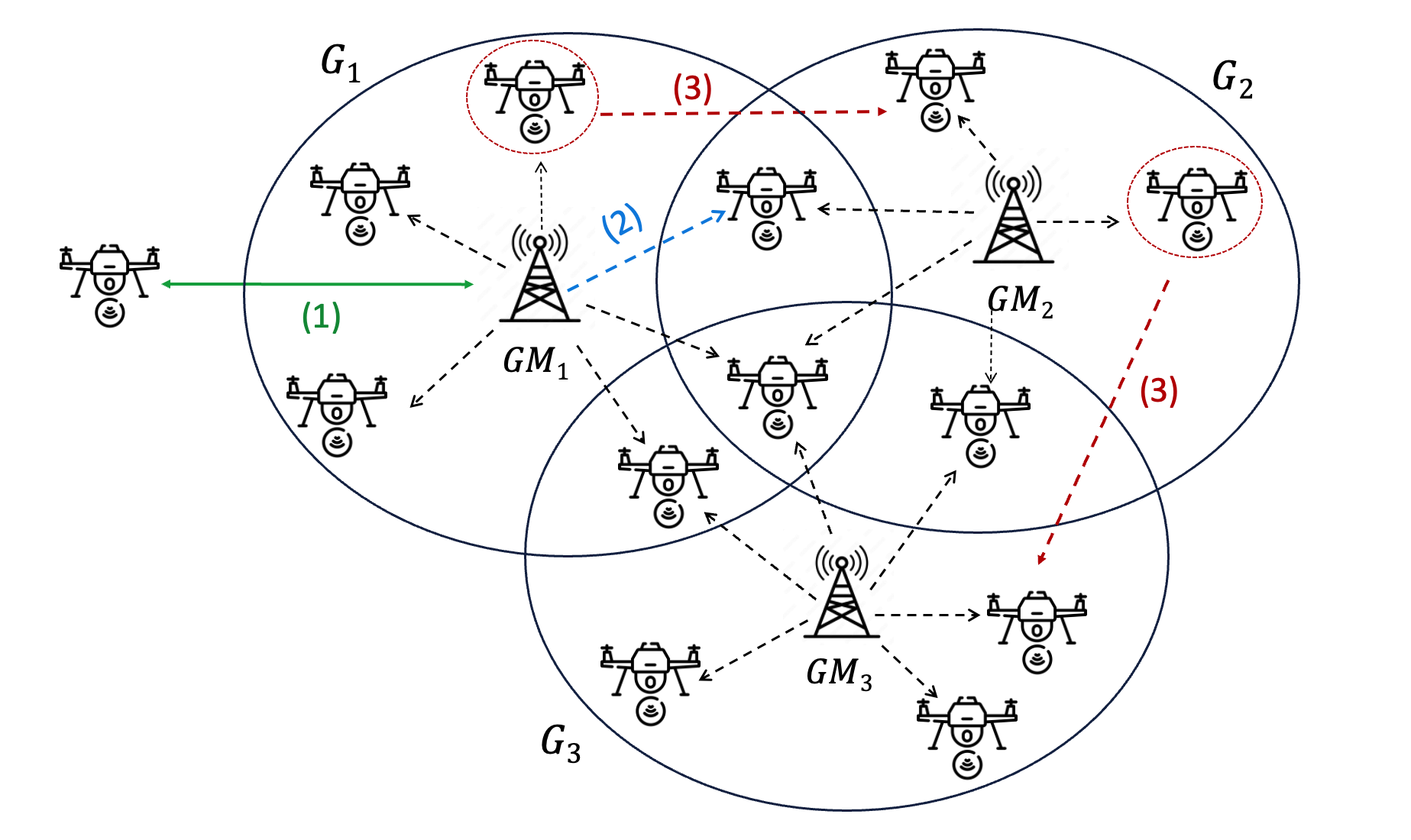}
	\caption{ The system model has three distinct connection types.  Type (1) represents the first registration via a secure channel. Type (2) demonstrates the channel between the group manager and the users who have completed the first registration phase and are ready to generate the secret key for confidential group communication and Type (3) illustrates the scenario where a user, who belongs to one group, joins another group. }\label{fig2} 
\end{figure}

 \indent Let consider a specific group  $G$  so that $U_{i}$ represent a member of it for some integer $i$.  A group manager, which is denoted by $GM$ for the group $G$ is assumed to have  more computational  power and energy resource compared to the other members of the group.  In  general, the key distribution of a user in the group is expected to be handled by the group manager. On the other hand, the proposed method also allows a way that each group member can authenticate others in $G$. All groups in the scheme employ  subspaces of a predetermined universal inner product space $E$. One might  select $E$ to be an infinite dimensional vector space; for example, $E$ can be all polynomials over a finite  field $\mathbb F$.  \\
\indent A basis $\mathcal B=\{v_{1},\dots,v_{n}\}$ for the subspace $W$  is selected to be secret to the group manager $GM$, that is,
$$W= \text{ Span }(\mathcal B)= \text{  Span }\{v_{1},\dots,v_{n}\}.$$
The nature of subspace $W$ allows one to select infinitely many bases for it but  knowing any basis for the subspace $W$ is sufficient to obtain the group’s secret key. However, the algorithm’s design makes it necessary to know the chosen basis to break the group’s authentication scheme. The group manager keeps $W$ and the selected basis $\mathcal B$ secret. The manager $GM$ also employs a randomly selected function $f(x)$ while  distributing the secrets to individuals. The selected function $f(x)$ can be a polynomial of the degree $d$. The security analysis in the next part indicates that the integer $d$ might be selected larger than the expected number of users in the group $G$. The group manager's secrets are shown in  Figure \ref{fig:secrets}. 

\begin{figure}[t] \label{FigGrp}
	
	\centering
	\includegraphics[width=\linewidth]{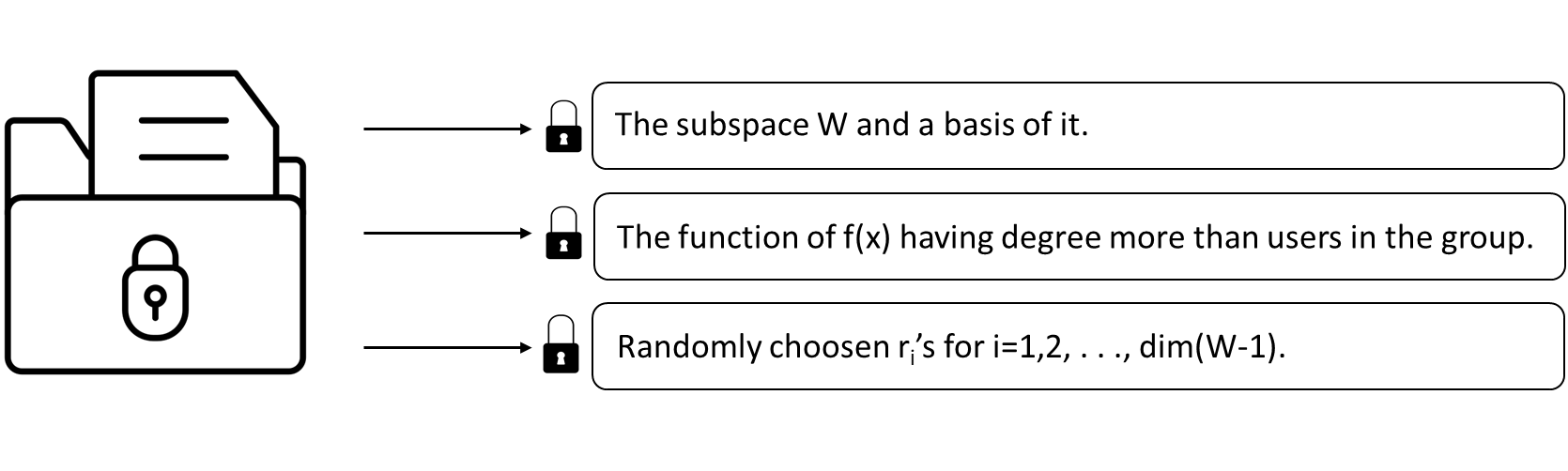}
	\caption{ The group manager's private information.} \label{fig:secrets}
\end{figure}

\indent Any user $U_{i}$ in the group $G$ is gets a public key $x_{i}$  which is preferably an integer. The secret of $U_{i}$ is $$\mathcal B_{i}=\{f(x_{i})v_{1},f(x_{i})r_{1}v_{2},f(x_{i})r_{2}v_{3},\dots,f(x_{i})r_{n-1}v_{n}\}$$ 
\noindent where $r_1,r_2,\dots,r_{n-1}$ are  numbers selected randomly by the corresponding group manager  $GM$ as depicted in Figure 2. These numbers stay the same for all members of the group. Note that each group member will have a linearly independent set in $W$, and these bases act as secret keys of the respected members. The private information of each individual is independent of one another. Algorithm \ref{OPAlg2}  shows how the group manager generates a key for each user.

\begin{algorithm}\label{OPAlg2}
	
	\DontPrintSemicolon
	\SetKwInput{kwInput}{Input}
	\SetKwInput{kwMain}{Key Derivation}
	\SetKwInput{kwOutput}{Key Distribution}
	\kwInput{ 
		\par
		\nl {Size of the group, degree of the function, Universal Space $E$.}\\
		
	}
	\kwMain{
		\par
		\nl $m\gets \text{  Size of the Group.}$\\
		\nl $d\gets \text{ Degree of the function } f(x)$.\\ 	
		\nl $\mathcal{B} \gets \text{ random } \{v_1,\dots,v_n\}: v_i\in E. $\\
		
		\nl $r \gets \text{ random } \{r_1,\dots,r_{n-1}\}: r_i \text{ in the base field } \mathbb F.$\\
		\nl $f(x)\gets \text { random element in } \mathbb F[x] \text{ of degree } d.$ 
		
	}
	\kwOutput{
		
		\nl Public key of user $U_i \gets x_i: x_i\in \mathbb F$. \\
		\nl Private key of user $U_i \gets \{f(x_i)v_1,r_1f(x_i)v_2,\dots,r_{n-1}f(x_i)v_n\}$

	\caption{Key Generation and Distribution: }
	}
\end{algorithm}

\subsubsection{ Group Authentication} The objective of  group authentication is to confirm whether several users belong to a group or not. This confirmation can be performed one by one but it would be infeasible in a crowded environment consisting of power constraint devices.  Our proposed method with inner product spaces addresses this issue in a practical way. In the proposed schemes, the registration phase assigns a basis to each member, and the private key of a member is embedded into the basis. In other words, each member has a distinct basis for a predetermined subspace $W$. Having a basis would allow one to confirm  its group membership. In the first step of the proposed algorithm, a group manager or any member of the group publishes  random  vectors $\mathfrak v$ and  $g$. The participants who are able to compute the inner product of $g$ with the projection of $\mathfrak v$, $Proj_W\mathfrak v$, would be confirmed.

\begin{algorithm}     \label{GrpAuth}   
	     
	\SetKwInput{kwInput}{Input}
	\SetKwInput{kwMain}{Secret Derivation}
	\SetKwInput{kwOutput}{Output}
	\kwInput{
		\par
		 The group manager or any member of the group publishes  a random vector $\mathfrak v$ in $E$ preferably  $ \mathfrak v\not\in W$ along with a nonce vector $g$
	\nl $\mathfrak v \gets \text { random element in } E$
	\nl $g \gets \text{ random element in } E$ 
		
	}
   \kwMain {
   	\par
   	 Each user $U$ computes $ \langle g,Proj_W \mathfrak v\rangle$ \\
   \nl $\mathfrak s \gets \langle g,Proj_W \mathfrak v\rangle$

 \kwOutput{
 	\par
 The user $U_i$ releases the requested bits of $\mathfrak s$.
}
	\caption{: Group Authentication}           
	}
	
\end{algorithm}

 Algorithm \ref{GrpAuth} generates a secret $\mathfrak s=\langle g,Proj_W \mathfrak v\rangle $ and this private number $\mathfrak s$ is employed  for authentication. Each member might be asked to release certain bits of $\mathfrak s$ to be  confirmed as a member of the group. The members can demonstrate their valid membership by revealing the required bits of $\mathfrak s$. The number $\mathfrak s$ has no other use. It is constructed with two random vectors, so each number is equally likely to be the secret. Hence, disclosing some bits of it does not reveal any information about the group’s private subspace. A similar method to establish a key for a confidential channel is described below.

\subsubsection{Establishing Group Secret Key}

The process of generating a group secret also assures that a  non-member can not join the group communication.  In other words, this proposed method eliminates the need for  extra authentication step in case a secret group key construction is needed for confidential communication. The group secret can be established via the direction of  the group manager or any trusted element in the group.  A member selects two random vectors $\mathfrak v \in E$ and $\mathfrak h \in E$
preferably $\mathfrak v,\mathfrak h \not \in W$. The vectors $\mathfrak v$ and $\mathfrak h$ are made public and the group secret $s$ is extracted from it by computing
$$s=\langle Proj_{W}\mathfrak v, \mathfrak h \rangle .$$

Note that computing $s$ requires information of any basis for the subspace $W$, in other words, the projection of $\mathfrak v$ onto the subspace of $W$ is the same regardless of using any basis for~$W$.

	\begin{algorithm}     \label{A1}    
		      
		\SetKwInput{kwInput}{Input}
		\SetKwInput{kwMain}{Secret Key Computation}
		
		\kwInput{
			\par
			The group manager or any member of the group publishes two random vectors  $\mathfrak v, \mathfrak h\in E$ 
			\nl $\mathfrak v \gets \text { random element in } E$
			\nl $\mathfrak h \gets \text{ random element in } E$ 
			
		}
		\kwMain {
			\par
			Each user $U$ computes $ \langle \mathfrak  h,Proj_W \mathfrak v\rangle$ \\
			\nl $ s \gets \langle \mathfrak  h,Proj_W \mathfrak v\rangle$
		}

	\caption{Key Establishment }

	\end{algorithm}

The secret $s$ then is utilized as a key of any symmetric key algorithm for  the confidentiality of messages.  Constructing a group secret key doesn't require any exchange of data among users in the group since each peer can obtain the key by using its own tokens, which are the  basis elements of $W$.  Unlike the other group authentication methods, every user in the group can extract the group key  without any other users' private information which eliminates the security concerns of users while sharing their secrets. To extract the secret s, each user only needs to perform one projection operation. This is different from other schemes that use polynomial interpolation, where the cost of operation depends on the number of users involved in the authentication process.\\
\indent In case two or more members in the group desire to set up a secure and confidential channel among themselves, they can use group key establishment scheme along with a classical method like Diffie-Hellman key exchange or a public key algorithm. In fact, in case of DH or a public key algorithm each user needs to share its public information with the others. However, in a group environment,  before sharing their public information, the users  first employ Key Establishment algorithm to set a secret "$s$" which can guarantee that only legitimate group members can join the process. Once such secret is established, they can employ Diffie-Hellman key exchange or any other public key  algorithm while sending their public information by encrypting the established secret $s$. Such method removes the necessity of a certification authority which is necessary in any public key environment to prevent the man in the middle attack.

\begin{algorithm}     \label{HybridKey}             
	
	\caption{:Confidential Channel Inside of a Group }   
	\begin{algorithmic}[1]

		\renewcommand{\algorithmicrequire}{\textbf{Input:}}
		
		\renewcommand{\algorithmicensure}{\textbf{Result:}}

			\STATE Two members $U_1$, $U_2$ establish a confidential channel.
			\STATE  They process Algorithm \ref{A1}
			to create a secret $s$.
			\STATE  $U_1$ and $U_2$ exchange their public information by encrypting them with the secret $s$.

	\end{algorithmic}
	
\end{algorithm}

\subsubsection{Adding a user to a group by a member}
The group manager $GM$ might not be available to handle adding new members to the group $G$.  In such cases,  a member should be able to add a non-member, denoted by $UF$, to the group and that user can communicate securely among the users in the group.  The group manager should be able to easily recognize the group member who added $UF$ to the group. $UF$ might or might not be given all privileges of the group member until it becomes a member via the group manager $GM$.
      
Consider the group member $U_{i}$, which has the following basis set:
$$B_{i}=\{f(x_{i})v_{1},f(x_{i})r_2v_{2},\dots,f(x_{i})r_nv_{n}\}$$
The $U_{i}$ is selects a random number $\mathfrak t$  and constructs a new basis for $W$:
$$B'_{i}=\{\mathfrak t f(x_{i})v_{1},\mathfrak t f(x_{i})r_2v_{2},\dots,\mathfrak t f(x_{i})r_nv_{n}\}$$
The new basis $B'_{i}$ is given to the new user $UF$. The public key of $UF$ is $$(x_{i},\ell)$$ where $x_{i}$ is the public key of $U_{i}$, and $\ell$ is the index selected by $U_{i}$ for $UF$.

The sponsor $U_{i}$ does not need to know the function $f(x)$ to add the user $UF$ to the group conversation. Note that the user $UF$ can easily grasp the group secret $s$ by using its basis and  the new user's sponsor can be recognized from its public key  by anyone or from its private key  by the group manager  having the function $f(x)$.

\section{ Security Analysis}
\subsection{Cryptanalysis}
The privacy of a single individual whose public information is $x_i$ is violated if an adversary has the information of the user's private key which is a basis set $\{f(x_i)v_1,\dots, f(x_i)r_nv_n\}$. In such a situation, the adversary has every privilege of the user. We are going to discuss if such a scenario is possible by exploiting broadcast public information. The members who are going to set up a secure channel agree on a vector $\mathfrak v$. Such a vector is made public. This public vector, $\mathfrak v$, is preferably  selected from outside of subspace $W$ which is generated by the users' bases. Another entity being used during an authentication and a key establishment process is a nonce vector $\mathfrak h$. This setup yields:\\
{\bf Public Information:} A random vector $\mathfrak v$ which is not in the subspace $W$ and a nonce vector  $\mathfrak h$. \\ 
The authentication and the key establishment processes requests users to compute the projection of $\mathfrak v$ on the $W$ and use this vector to compute the inner product with the nonce $\mathfrak h$.

\begin{proposition} \label{Prp1}
	Let $\mathfrak v$ be a vector in a universal space $E$ and $t$ be its projection vector onto a subspace $W$ of $E$. Let $\mathfrak h$ be a random vector and $$\mathfrak p= \langle t, \mathfrak h\rangle$$
	It is not feasible to obtain any information about the subspace $W$ from $\mathfrak v, \mathfrak h$.   
\end{proposition}

	\begin{theorem}
 Let $\mathfrak v, E, t, W$ and $\mathfrak h$ be as above.	The number	$\mathfrak p$ is a session secret for authentication and is secret to the members joining the authentication phase. Each subspace of $E$ is equally likely to be $W$ even if one obtains the secret $\mathfrak p$.
	\end{theorem}

We present proofs of above two claims in the following.

\begin{proof}
	The $\mathfrak h$ and  $\mathfrak v$ are random ephemeral vectors. As they are random in the universal space, they have no relation with the subspace $W$. The randomness of the selection process indicate that the probability that any of these two vectors is in $W$ is negligible. Therefore, we might assume that $\mathfrak h$ and $\mathfrak v$ are vectors outside of $W$. \\
	\indent On the other hand, let's assume for a moment that an adversary model such as a game based adversary model is able to lead a method where the probability that the revealing of the secret $\mathfrak p$ is non-negligible. In other words, an adversary is assumed to obtain the session key $\mathfrak p$ for several session. The key is an element of the field that the universal space $E$ lies on it. The secret $\mathfrak p$ is obtained from computing $\langle t, \mathfrak h\rangle$ where $\mathfrak h$ is public but $t$ is not.\\
	{\it Claim:} In any subspace $B$ of $E$, one can find a vector $b$ such that $$\langle b, \mathfrak h \rangle =\mathfrak p$$ 
	Let $B$ be a subspace of $E$ and $\mathfrak b$ be a vector in $B$. Suppose $$\langle \mathfrak b, \mathfrak h\rangle =\mathfrak c$$ 
	
	\noindent and let $$\mathfrak a =\dfrac{\mathfrak p}{\mathfrak c}.$$ 
	
	\noindent Then the vector $b=\mathfrak c \mathfrak b$ in $B$ gives 
	$$\langle b, \mathfrak h\rangle =\mathfrak p.$$

	The above argument indicates that any subspace of $E$ is equally likely to be $W$. In other words, even if an adversary obtains the secret $\mathfrak p$, the adversary cannot gain any clue about the secret subspace $W$ since there exists a vector $b$ for any subspace of $E$ such that$\langle b, \mathfrak h \rangle =\mathfrak p$.

\end{proof}

\begin{remark}
The above statements imply that the secrecy of the subspace $W$ of $E$ is preserved even if an adversary obtains the session keys multiple times. In fact, the adversary cannot infer anything about the dimension of $W$, let alone the subspace itself.
\end{remark}

In the key establishment phase, the constructed private key $\mathfrak s$ is going to be used during secure group communication. The adversary  has  only the public vector $\mathfrak v$. The following proposition claims that it is not feasible to get this vector without knowing the subspace $W$.

\begin{proposition} \label{Prop2}
	Let $G$ be a group and $\mathfrak v_t$ be a public vector broadcast at the time $t$. Suppose that an adversary captures  a number of distinct public vectors at different times. It is still infeasible for the adversary to guess  the required basis for the group's secret key let alone a legitimate user's private information. 	
	
\end{proposition}
\begin{proof}
	The public vectors $\mathfrak v_t$ are selected by members of the group and they are not in $W$. As $W$ sits on a universal space $E$ which might be of infinite dimension, knowing several vectors outside of $W$ is not enough to generate $W$ itself.
\end{proof}

\begin{proposition}
	Suppose an adversary $A$ has the basis set $\{f(x_i)v_1,\dots, f(x_i)r_nv_n\}$ which belongs to a user $U_i$. It is still infeasible to obtain $f(x_i)$. 
\end{proposition}

\begin{proof}
	The adversary does not know the original basis set $\{v_1,\dots,v_n\}$ and the scalars  $r_2,\dots,r_n$ therefore it is not possible to guess $f(x_i)$.   	
\end{proof}

\begin{proposition}\label{FFsize}
	Suppose an adversary $A$ has obtained more than one user's bases. It is still not possible to obtain the function $f(x)$ or any other user's private information.
\end{proposition}

\begin{proof}
	At this point, we should note there that in practice the inner product spaces are taken to be over large-size finite or real fields. Let's assume that the adversary $A$ has the following bases:
	$$
	\{f(x_i)v_1,\dots, f(x_i)r_nv_n\},\cdots,\{f(x_{j})v_1,\dots,f(x_j)r_nv_n\}
	$$	
	From this information, it is possible to obtain the ratio of the value of public information of users under the function $f(x)$ but not the function itself. Without the function $f(x)$, it is not possible to get other members' private data.	
\end{proof}

\begin{proposition}\label{Prop1}
	Let $S_i=\{f(x_i)b_1,r_2b_2,\dots,r_nb_n\}$ be a basis for a subspace $W$ of vector space $E$ for $i=1,\dots,\mathfrak k$ and $f(x)$ is  function of degree $d$. As long as $d>\mathfrak k$, it is infeasible to construct $f(x)$. 	
\end{proposition} 

\begin{proof}
	The polynomial $f(x)$ has degree $d$ and it is known from Newton's theorem that constructing $f(x)$ is only possible when at least $d+1$ points on the graph of $f(x)$ is known.	
\end{proof}

Proposition \ref{FFsize} forces to define a security parameter which is the size of a field $\mathbb F$ where the employed inner product space lies. In addition to this, the degree of $f(x)$ is another security parameter by Proposition \ref{Prop1}. Let the security parameters be
$$\beta_1=\text { size of } \mathbb F \quad \text { and  } \beta_2=\deg f(x)$$

It is not hard to observe that these parameters are directly related to the cost of key derivation operations. In other words, they offer a trade-off between security\&privacy and the cost of computations. 

\subsection {Known Cyber Attack Analysis: }
In this section, we  discuss the well-known thread models and the proposed algorithm's resistance to these attacks.\\
\indent \textbf{DOS Attack:} In the secret sharing-based GASs, the authentication is performed only when  a certain number of shares or more are available and the authentication can be completed only if all participants are legitimate.   This means that even if a single illegitimate user participates in the authentication process, the authentication fails. In addition to this, the group manager or members joining the authentication process can not recognize the illegitimate user and  this makes such algorithms  vulnerable to DOS attacks. This means that any attacker can disrupt the service and the culprit may be unknown. In our scheme, only those who have a basis set for a fixed subspace can join the authentication or key agreement phase. Thus, an intruder with an invalid basis cannot interfere with the authentication phase or cause any trouble. Therefore, our scheme prevents a DOS attack by an attacker.\\
\indent \textbf{ Replay and Man-in-the-Middle Attack:}  The fact that the authentication algorithm does not require any private data sharing creates a safe environment against Man-in-The-Middle Attack. In fact, since an authentication or construction of a shared secret key is handled privately by using a public random vector, an intruder without knowledge of the predetermined subspace, can not proceed communicating with any group member.    Moreover, since for each authentication session, the  group publishes a different vector and a nonce, any adversary sniffing exchanged data in earlier sessions can not perform a replay attack.

\indent \textbf{Group Manager Compromise Attack: }A node in a group has two sets of crucial data: One is for confidential group communication, and the other one is for confirming its identity in the group. This crucial data is provided by the manager of the group. Therefore, compromising the group manager's secret allows one to access all members' private information.  Unlike the other group authentication algorithms that are based on polynomial interpolation, compromising a few members' private keys does not allow an adversary to impersonate a group manager. As noted in Proposition  \ref {Prop1}, it is not possible to take over the group management unless the attacker can capture at least $d+1$ points on the graph of  $f(x)$  where $d$ is the degree of the  group manager's private function $f(x)$.

Let $GM$ be the manager of the group $G$ and $f(x)$ be the private data of $GM$ such that $\deg f(x)=d$. Assume that an adversary obtains the secrets of some members of $G$. Let $S_i= \{f(x_i)b_1,f(x_i)r_2b_2,\dots,f(x_i)r_nb_n\}$ be the such private information for the users $U_i$ for $i=1,\dots,\mathfrak t$.  The knowledge of this information does not allow for the extraction of the point $(x_j,f(x_j))$ on the graph of  $f(x)$. Even if such an information is obtained,  as long as $\mathfrak t<d$, the adversary can not act as a group manager  by Proposition \ref{Prop1}.\\
\indent \textbf{ Impersonate group membership:}
Each member has a basis for $W$, and the security of group communication depends on the knowledge of the subspace $W$. In other words, an adversary does not need to know individuals' secrets to compromise the confidentiality of group communication; instead, any basis for $W$ provides enough tools for an adversary. Therefore, the only way to impersonate a group member is to have a basis for $W$. The above propositions \ref{Prp1}, \ref{Prop2} say  it is not feasible to obtain a set that can span $W$ from public information disclosed during authentication or key establishment sessions. We show that obtaining the session key during an authentication or key establishment session can be used as a subroutine to get a basis for $W$. 
\begin{theorem}
	Let $W$ be a subspace of $E$ with a dimension $d$. If an adversary is able to extract a session secret for authentication or key establishment, then the procedure of the adversary can be used to reveal the subspace $W$. 
\end{theorem}

\begin{proof}
	Let's assume the adversary $A$ infiltrates a group $G$. Let $W$ be the secret subspace of $G$ with dimension $d$. In order to be able to join confidential communication, $A$ needs to compute the projection of random vectors onto $W$. In fact, in a communication session, an ephemeral random vector $\mathfrak v$ is released, and each member computing $s=Proc_W\mathfrak v$ can join the conversation. If $A$ can impersonate a group member then $A$ must have the projection vector $s$. The projection vector $s$ lies in the space $W$. If $A$ can attend $d$ number of sessions, then $A$ has $d$ projection of random vectors onto $W$. The probability of $d$ vectors being dependent is negligible, and as the dimension of $W$ is $d$, such $d$ vectors generate the whole subspace $W$.
\end{proof} 	
\begin{corollary}
	The above discussion shows that an adversary can impersonate a group member if and only if the adversary has a basis for $W$. 
\end{corollary}
\indent \textbf{ Impersonate a specific member:}	The identity confirmation takes place while authentication occurs with the group manager. In such cases, the private information, which is the image of the public key of the user under the group manager secret function, needs to be confirmed by the manager. In other words, in addition to having a basis for the subspace $W$, an adversary also needs to have the secret function of the manager.  \\
\indent In order to illustrate the impersonation  attacks, below we present a simple example.

\begin{example}
	Let denote the set of users by $G$ and assume that the universal space is $E=\mathbb R^{10}$ and the inner product is the well-known dot product.  In practical applications, $E$ can also be selected as polynomial space over a field with the standard inner product involving integral. Consider a subspace $W$ of dimension 3. The group manager selects a basis $\{b_1,b_2,b_3\}$ for $W$ and a random polynomial $f(x)$ of degree 10. Note that each $b_i$ lies in $\mathbb R^{10}$. The public key of the user $U_i$ is $i$ and the private key is $$U_i\leftarrow\left\{f(i)b_1,f(i)\dfrac{1}{3}b_2,f(i)\dfrac{2}{5}b_3\right \}$$ 
	Note that the group manager selects $r_2=1/3$ and $r_2=2/5$. 
	In order to involve confirmation and key establishment phase, each member $U_i$ needs to compute the projection of a random vector $\mathfrak v$ in $\mathbb R^{10}$.\\
	The public information is just the vector $\mathfrak v$. \\
	Therefore an adversary only knows some information about the universal space $E=\mathbb R^{10}$. \\
	The secret of the group is $W$ which is generated by $\{b_1,b_2,b_3\}$.  There are infinitely many subspaces of $\mathbb R^{10}$ of dimension three. Note that the dimension is also secret to the group members.\\
	An adversary who wants  to join the group confidential communication needs to find the projection of $\mathfrak v$ onto $W$ where he/she only knows $\mathfrak v$.
\end{example}

\indent {\bf Forward Secrecy}:   In the first version of the group authentication schemes  \cite{Harn}, the users are able to use the keys only once. Therefore, there is no permanent key involved and the perfect forward secrecy is provided. Harn's third method and the method with elliptic curves allow the user to employ their keys multiple times, which means that users have permanent keys. However, to extract the key to providing confidential group communication, one needs several other users' public and private information. Unless one obtains at least $t$ members' information (assuming $(t,m,n)$ is in place), it is not possible to recover the group's secret key and reveal the group's earlier exchanged data. \\
The group secret in our authentication  is extracted by following the steps:
\begin{enumerate}
	\item A group manager or any member broadcasts two  ephemeral random vectors $\mathfrak v,g$ selected from the universal space $E$.
	\item Each member $U$ in the group computes projection $Proj_W\mathfrak v$ of $\mathfrak v$ onto the space $W$ where each member has a basis for it.
	\item Then each user computes the inner product $$\langle g, Proj_W\mathfrak v\rangle$$
	and extract the key for authentication. The users exchange predetermined bits of this secret to be confirmed by others in the group. 
\end{enumerate}  
Note that the use of a nonce in the authentication phase is necessary as some information about the resulting value is shared by the members among them in an open channel. In such a situation, if $g$ is not placed as in the case of the secret key establishment, the users should exchange some information about the projection of $\mathfrak v$ onto $W$. As the projection vector lies in $W$, an adversary can obtain a basis by sniffing data from several authentication sessions. 

On the other hand, in the key establishment scheme, each user employs a random vector for each session to extract the secret. In other words, the secret key for confidential communication is obtained via a permanent key and a temporary session vector.  Each group member can use its secret multiple times; in other words, the public and private information of users can be considered permanent keys. However, for each group communication session, the secret is established through the use of an ephemeral key $v$. The extracting of the session is not only depending on the permanent key. The ephemeral key is also needed to be present to reveal the group's secret key for each session. In other words, the perfect forward secrecy is provided by putting ephemeral keys into the play of extracting the group session keys.

\section{Run Time }
During an authentication or key establishment phase, any user needs to compute the projection of a given vector onto the subspace generated by the basis elements. Each user has a unique basis  for the subspace determined by the group manager. Let $W$ be the assigned subspace by the group manager $GM$ such that $$\dim W=n.$$

Each member has a basis consisting of exactly $n$ elements.  Each user joining an authentication or key establishment process needs to compute the projection of a vector on the $W$. One way to find the projection vector requires obtaining an orthonormal basis for $W$. Such a basis can be found by  using the Gram-Schmidt orthogonalization procedure which  requires
$$\sum_{i=1}^n 2i-2=n^2-n \text{ inner products}. $$
The normalization of the basis elements needs $n$ more inner products. Let $\{b_1,\dots,b_n\}$ be an orthonormal basis for $W$. The projection $\mathfrak g$ of $\mathfrak v$ onto $W$ can be computed as
$$\mathfrak g=\langle \mathfrak v,b_1\rangle b_1+\langle \mathfrak v,b_2\rangle b_2+\dots+\langle \mathfrak v,b_n\rangle b_n.$$

Overall, the number of inner products for computing projection of $\mathfrak v$ onto $W$ is bounded by $O(n^2)$ where again $n$ stands for the dimension of $W$. Since the cost of the overall operation for producing the projection vector is dominated by the inner product computation, the running time of the algorithm is bounded by $O(n^2)$ inner products. \\
\indent The second way to compute projection of a vector $\mathfrak v$ on $W$ is the following:
Let $\mathfrak g=a_1b_1+\dots+a_nb_n$ be the projection of $\mathfrak v$ onto $W$ where $a_i$ lies in the base field $\in \mathbb F$. The coefficients $a_i\in \mathbb F$ can be found by solving the linear system:
$$\begin{array}{ccccccc}
	\langle \mathfrak v, b_1\rangle& =&\langle \mathfrak g,b_1 \rangle&=& \langle a_1b_1+\dots+a_nb_n ,b_1 \rangle\\ 
	\langle \mathfrak v, b_2\rangle& =&\langle \mathfrak g,b_2 \rangle&=& \langle a_1b_1+\dots+a_nb_n ,b_2 \rangle\\ 
	\cdot & \cdot & \cdot \\
	\cdot & \cdot & \cdot \\
	\cdot & \cdot & \cdot \\
	\langle \mathfrak v, b_n\rangle& =&\langle \mathfrak g,b_n \rangle&=& \langle a_1b_1+\dots+a_nb_n ,b_n \rangle.\\ 
\end{array}$$

The  number of inner products on the right side of the above system is $$n+(n-1)+(n-2)+\dots+2+1=\dfrac{n^2+n}{2}$$
\indent Overall, $(n^2+3n)/2$ inner product computations should be performed. Even though the number of inner products is asymptotically again $O(n^2)$ in this direction, the real-time implementation might be more efficient in certain cases. In fact, the computations to solve linear systems can be performed  in parallel which allows exploiting a multi-core environment. As most of the systems, including IoT devices, use multi-core processors to their hardware architecture, solving linear systems instead of applying the Gram-Schmidt orthogonalization process might be less costly for some devices. Compared to other group authentication schemes, the proposed method of authentication and key exchange methods do not depend on the number of users in the groups. As mentioned above, group authentication schemes can be classified into two generations based on the mathematical tools they use. The first generation schemes use polynomial interpolation and multiplicative subgroups of finite fields. The second generation schemes use polynomial interpolation and elliptic curves over finite fields. Elliptic curve cryptography offers higher security and efficiency than finite field cryptography. The first generation group authentication scheme starts with interpolating points received from the users joining an authentication process. If the group manager employs a polynomial $f(x)$ of degree $n$, the number of users in an authentication process should be at least  $n+1$. The number of users in an authentication process  might be much larger than $n+1$. Let $k$ be the number of users joining the authentication phase then each user needs to perform at least $O(k^2)$ operation for the polynomial interpolation where $k>n$. The second generation group authentication schemes needs to perform group operation in an elliptic curve over a finite field. In other words, each user $U_i$ needs to compute $$C_i=\left(\prod_{r=1,r\not=i}^{k} \frac{-x_r}{x_i - x_r}\right)f(x_i)P.$$ In the first part, each user performs $k$ inversion and multiplication operation in the finite field, $\mathbb F_p$, then needs to find a multiple of a point $P$ in the elliptic curve $E(\mathbb F_p)$. Again, as above the number $k$ should always be larger than the degree of $f(x)$. Therefore, each user performs at least $O(k)$ operation including inversion and multiplication before computing a multiple of $P \in E(\mathbb F_p)$. Then each user needs sum up the points received from others joining the authentication phase. The overall cost is then $O(ck\log^3 p)$ where $c$ is a constant larger than $\log p$. We should note here than an addition in the elliptic curve group, $E(\mathbb F_p)$ takes about $O(\log^3 p)$ operation \cite{Frey}. The proposed method requires each member performing projection operation regardless of number of users joining the authentication phase. The overall cost to each user joining the phase is about $O(n^2)$ where $n$ is the dimension of the selected subspace and $n$ can be chosen a small number to reduce the cost in an authentication operation.  \\
\indent  An additional cost factor for group authentication schemes is the number of communication channels that must be established between each peer and the others during the authentication or key establishment phases. In the scenario we propose, either a member or the group manager disseminates a single piece of data, which each incoming member utilizes to authenticate their identity or derive the group’s secret. Conversely, authentication schemes that rely on polynomial interpolation necessitate that each user acquire the computational outcomes from all other users. Essentially, these schemes mandate that every member access the public data of their peers, execute calculations, and disseminate their computational findings to every new participant in the session. Table \ref{demo-table} provides this overview, where it presents the number of communication channels needed to be established for each member when $m$ number of members are joining the session.

\begin{table}[tb]
	\caption{ The number of communication channels that a single peer must establish when $m$ peers are active in the authentication or key establishment phase.}
	\begin{center}
		\begin{tabular}{ | p{1.8cm}  | p{1.8cm}  |p{1.8cm}  | p{1.8cm} |}
			\hline  
			Phase   &  1st Gen &   2nd Gen &   Proposed \\ \hline
			Key establishment &   $2(m-1)$  &   $2(m-1)$ &   $1$  \\ \hline
			Authentication  &   $2(m-1)$  &   $2(m-1)$ &   $1$  \\ \hline
		\end{tabular}\label{demo-table} 
	\end{center}
\end{table}

We have conducted real-time tests to give an insight into the practicality of our authentication scheme. The real-time tests also include first and second-generation group authentication schemes.    Real-time testing was performed on two distinct systems. The initial series of tests were executed on a computer equipped with the Linux operating system, powered by an Intel i7-11370H processor and bolstered by 32 GB of RAM. In parallel, we have conducted a comprehensive series of tests, expanding the dimension options on a Raspberry Pi 4B. This device is equipped with 8 GB of RAM and a 1.5 GHz Quad-Core 64-bit Arm Cortex A72 CPU, providing robust performance for our computational needs.

\begin{figure}[H] \label{graph1}
	\centering
	\begin{tabular}{c}
		\includegraphics[width=0.37\textwidth]{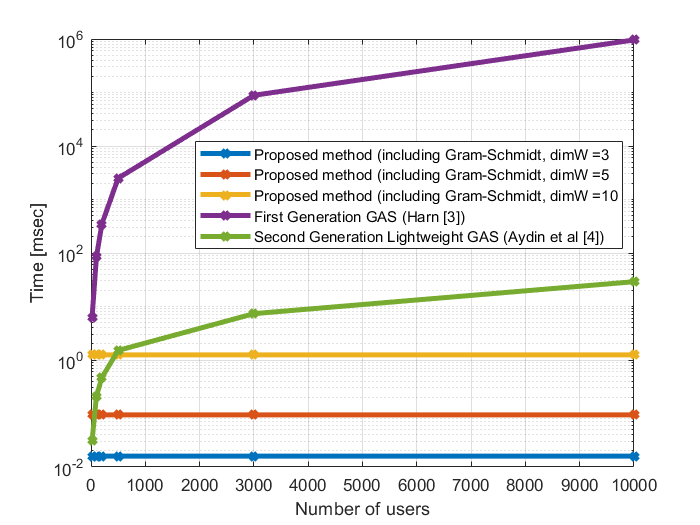}\\[\abovecaptionskip]
		\footnotesize Figure 4a. Results on PC
	\end{tabular}
		
 	\begin{tabular}{c}
 	\includegraphics[width=0.39\textwidth]{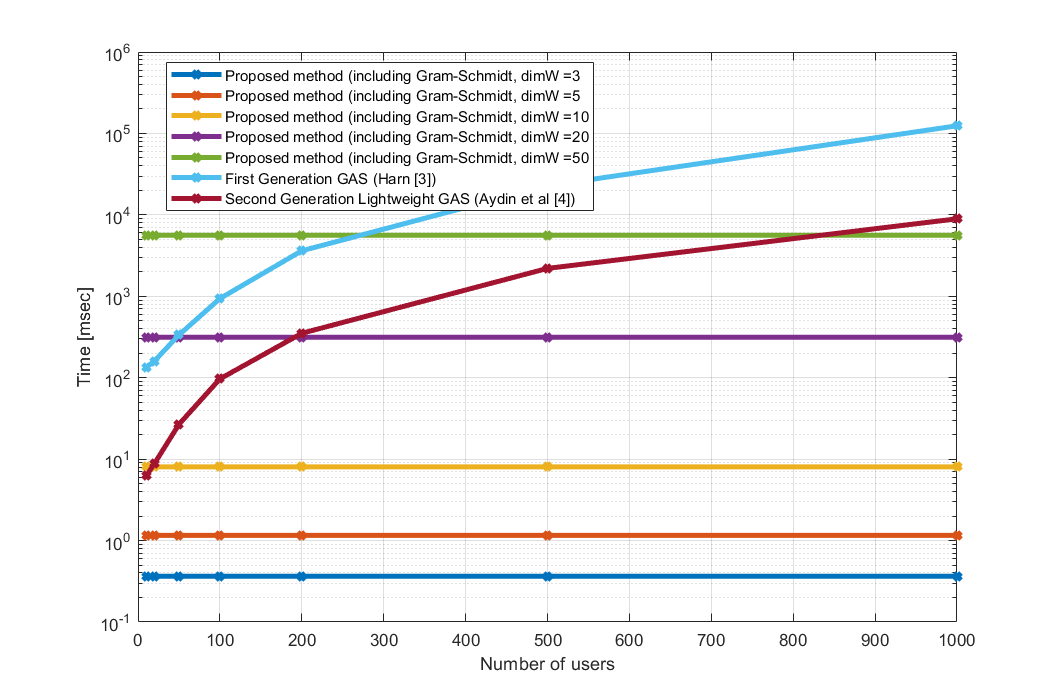}\\[\abovecaptionskip]
 	\footnotesize Figure 4b. Results on RP 4B
 \end{tabular}

	\caption{Comparison of the proposed method with the first and the second generation GASs. The graph represents the computational cost for each member, excluding communication costs. The top graph displays the outcomes of the test conducted on a computer, while the subsequent graph presents the results for the Raspberry Pi 4 Model B. }
\end{figure}

 Group members are required to derive an orthonormal basis from their designated keys, which will serve as a foundation for any subsequent authentication or key establishment protocols. Therefore, we have also performed tests that omit the process of extracting an orthonormal basis. The outcomes of these tests are systematically presented in Figure 5.

\begin{figure}[H] \label{graph2}

\centering
\begin{tabular}{c}
	\includegraphics[width=0.37\textwidth]{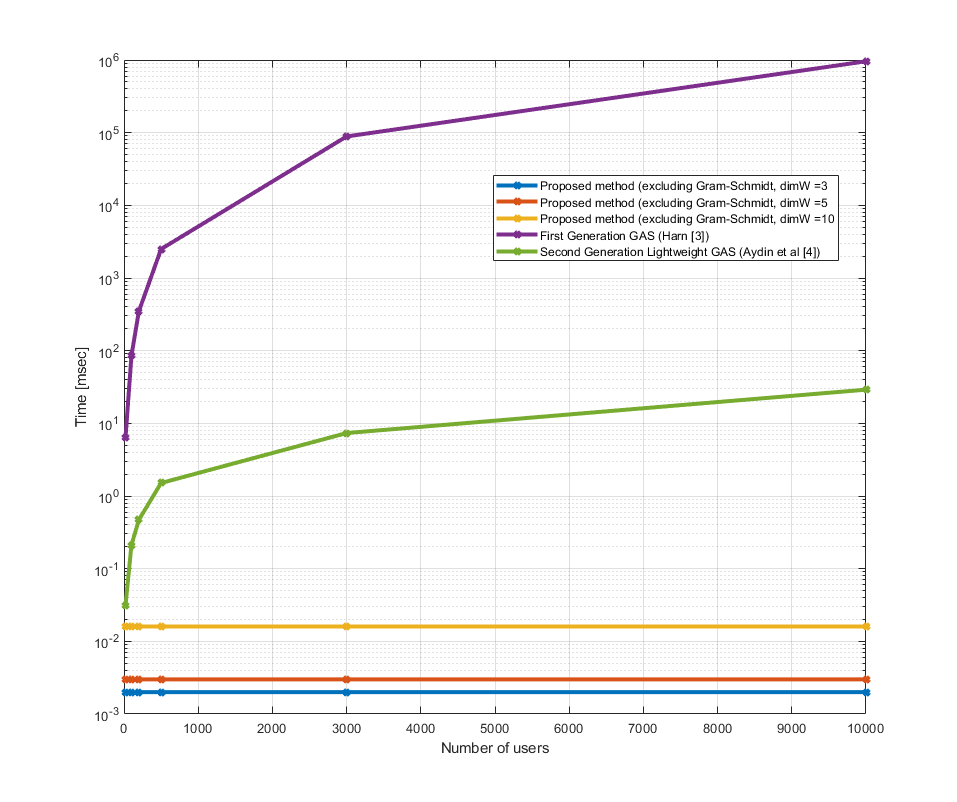}\\[\abovecaptionskip]
	\footnotesize Figure 5a. Results on PC
\end{tabular}

\begin{tabular}{c}
	\includegraphics[width=0.39\textwidth]{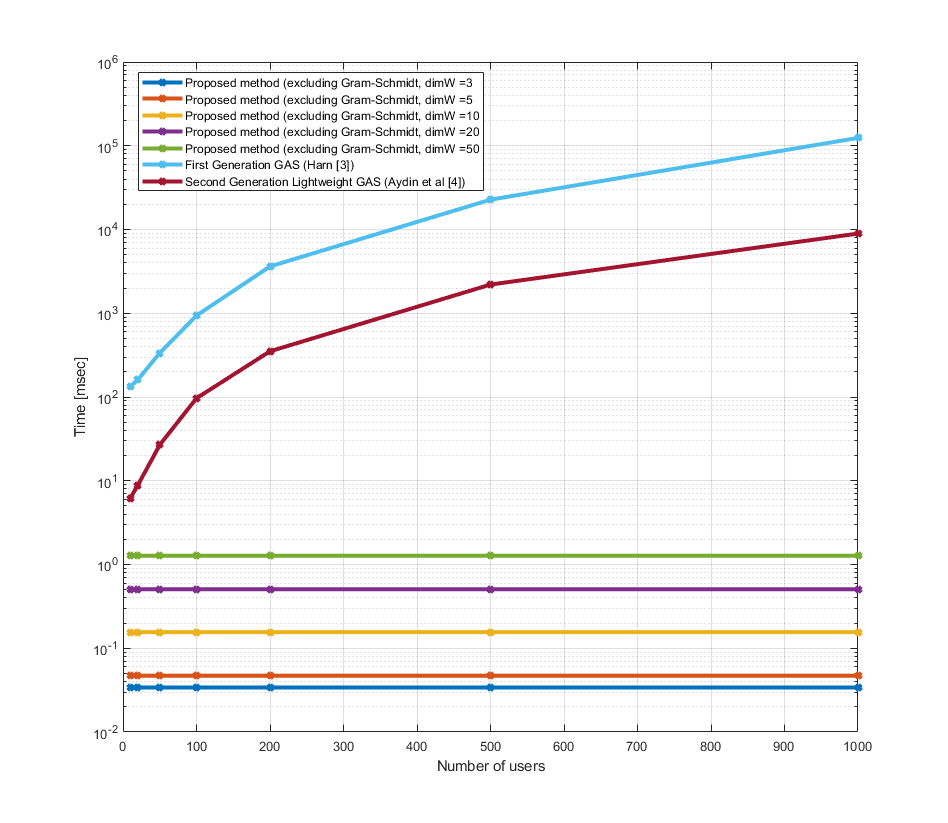}\\[\abovecaptionskip]
	\footnotesize Figure 5b. Results on RP 4B
\end{tabular}
	\caption{Comparison of the proposed method with Harn's and the second generation GASs.  The graph represents the computational cost for each member, excluding communication costs. The computational cost of the proposed algorithm excludes GM orthoganalization process as it is enough to perform it only one time. Again, the top graph displays the outcomes of the test conducted on a computer, while the subsequent graph presents the results for the Raspberry Pi 4 Model B }
\end{figure}

\indent A scheme described in the first and generation group authentication schemes are classified as $(t,m,n)$ scheme where
\begin{enumerate}
	\item $t$ is the minimum number of users required in an authentication process.
	\item   $m\ge t$ is the number of users joining the authentication process
	\item $n$ is the number of users in the group.
\end{enumerate} 
We observed that the running time is independent of the number $t$, which is the degree of the group manager’s secret function,  for the methods based on polynomial interpolation. However, the security parameter is proportional to $t$. The running time of both the first and second-generation schemes is linear in the number of users $m$ who participate in the authentication process. We implemented all algorithms in the same environment where we use a C++ library, PARI/GP \cite{PARI}, for operations involving high-precision numbers. We also note here that our scheme's cost to users is independent of the number of users joining the authentication and key establishment process and we select $W$ as a subspace of $\mathbb R^{n}$ where the inner product is just the dot product.   For the first generation GAS algorithm, we use a prime field of size 2048 bits, whereas, for the second generation, we use a prime field of size 224 bits. In this case, both schemes give almost the same security level, which is about 112 bits, as the first generation is based on the discrete logarithm problem for a multiplicative subgroup of a finite field, whereas the second generation is based on the discrete logarithm problem on an elliptic curve group. Besides having a nearly twofold running time advantage, the second generation algorithms also benefit from performing operations in a smaller field size, which may offer a significant memory advantage.

\section{Conclusion and  Future Work}
\subsection{Drawbacks}
The proposed method is suitable for various applications, such as swarms of drones in the military or civic sector. The scheme is the first approach that uses inner product spaces for authentication and key establishment among a group of devices. The main motivation to seek a new mathematical tool for this construction is to overcome the issue inherent to polynomial interpolation. That is, if a single device acts maliciously, it can disrupt either the authentication or the key establishment process. The proposed scheme does not have this issue, but it imposes a memory burden on each member of the group. That is, the secret of each member has a larger size compared to members in other group authentication schemes. In addition to this, it has also some  limitations that can be improved:
	\begin{enumerate}
	\item A subgroup of users may require a confidential channel among themselves. This can only be achieved by combining the proposed group authentication and a classical key establishment method, such as the Diffie-Hellman key exchange.
	\item  Within a group, a member may cause harm by adding unauthorized users. This can be achieved by constructing a random basis for the group subspace and distributing it to these unauthorized users. As a result, these users can seamlessly join group conversations, and their intrusion can only be identified by the group manager.
	\item The group members may need to revoke the membership of malicious users. This can only be done by re-generating the keys and re-establishing the group.
	\item  The proposed scheme enables any member to add a non-member to the group. This is based on the assumption that group membership is permanent, as in the scenario of drones conducting surveillance. However, revocation may be necessary in different situations and should be addressed in the future work of group authentication. On the other hand, at this stage the only way to remove the member credentials is re-distributing new credentials or establishing a new group.
\end{enumerate}

\subsection{Future Plan}
Group authentication is a novel idea that has the potential to be applied to various environments and applications. The proposed group authentication method is computationally lightweight and enables any member to authenticate other group members and establish a secret key among them. The proposed algorithm uses inner products to create a new group authentication method that aims to provide both the desired security level and a low computational load.

As the first group authentication algorithm with inner product spaces, the method can be further improved to provide group handover schemes for near-future mobile base stations. Considering the near-future IoT networks, which are expected to be autonomous and decentralized, each user should be able to introduce their trusted partners to the groups they belong to. In this respect, the research to create a secure and reliable system should take into account such a demand. In some cases, a new member who is not introduced by the group manager may need to be excluded from some group conversations until they are registered by the group manager, and future work should aim to design such an algorithm. 

The group authentication and key establishment schemes are suitable for decentralized networks. In these networks, some group members may need to exchange confidential data among themselves, where both the group membership authentication and the individual authentication are necessary. Such flexibility should be incorporated into the current method in future work. While conducting real-time tests, we also observe that there is room to improve the first and second-generation schemes. Future work that focuses on reducing the real-time cost should combine the mathematical tools in all three generations of group authentication schemes.

\section*{Biography Section}
\begin{IEEEbiography}[{\includegraphics[width=1in,height=1.25in,clip,keepaspectratio]{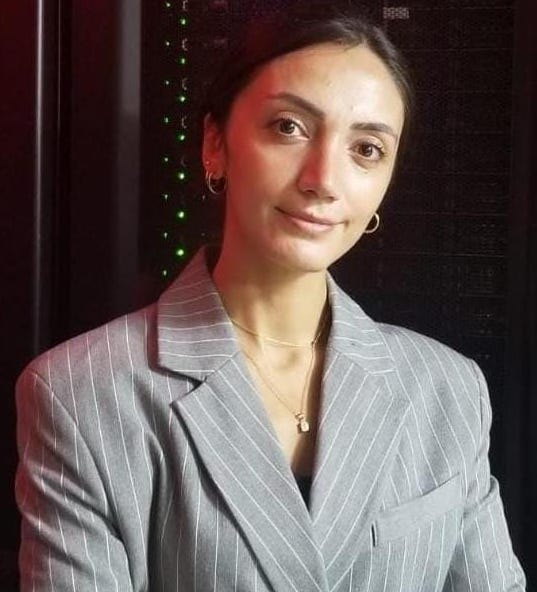}}]{Sueda Guzey} 
	(Student Member, IEEE) received BS in Mathematics from Middle East Technical University, Turkey in 2017. She is currently pursuing a PhD degree in the Cyber Security Engineering and Cryptography Program at Institute of Informatics, Istanbul Technical University.   Her current research interests include Cryptography, Cyber Security, Information Security, Number Theory, and Network Security.
	
\end{IEEEbiography}

\begin{IEEEbiography}[{\includegraphics[width=1in,height=1.25in,clip,keepaspectratio]{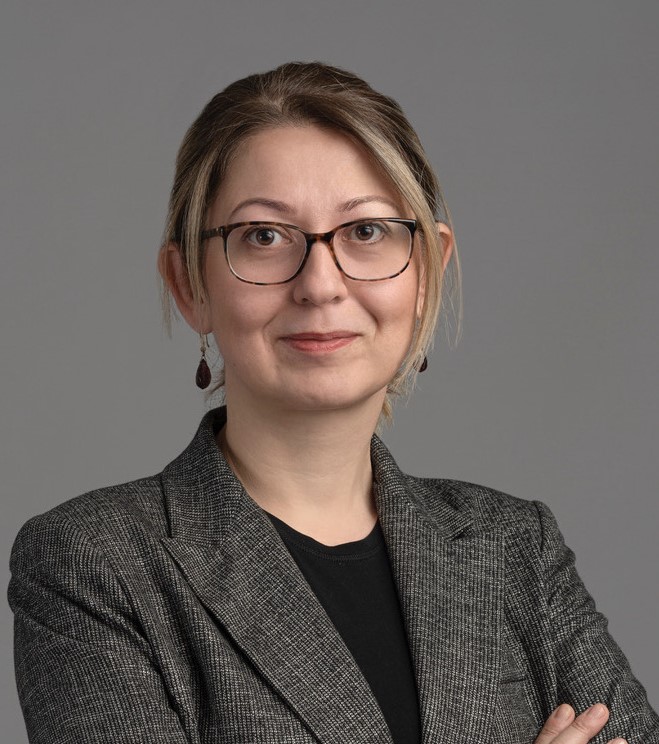}}]{Gunes Karabulut Kurt} is currently an Associate Professor of Electrical Engineering at Polytechnique Montréal, Montréal, QC, Canada. She is a Marie Curie Fellow and has received the Turkish Academy of Sciences Outstanding Young Scientist (TÜBA-GEBIP) Award in 2019. In addition, she is an adjunct research professor at Carleton University. She is also currently serving as an Associate Technical Editor (ATE) of the IEEE Communications Magazine and a member of the IEEE WCNC Steering Board. She is the chair of the IEEE special interest group entitled “Satellite Mega-constellations: Communications and Networking”. Her research interests include space information networks, satellite networking, wireless network coding, wireless security,  space security, and wireless testbeds.
	
\end{IEEEbiography}

\begin{IEEEbiography}[{\includegraphics[width=1in,height=1.25in,clip,keepaspectratio]{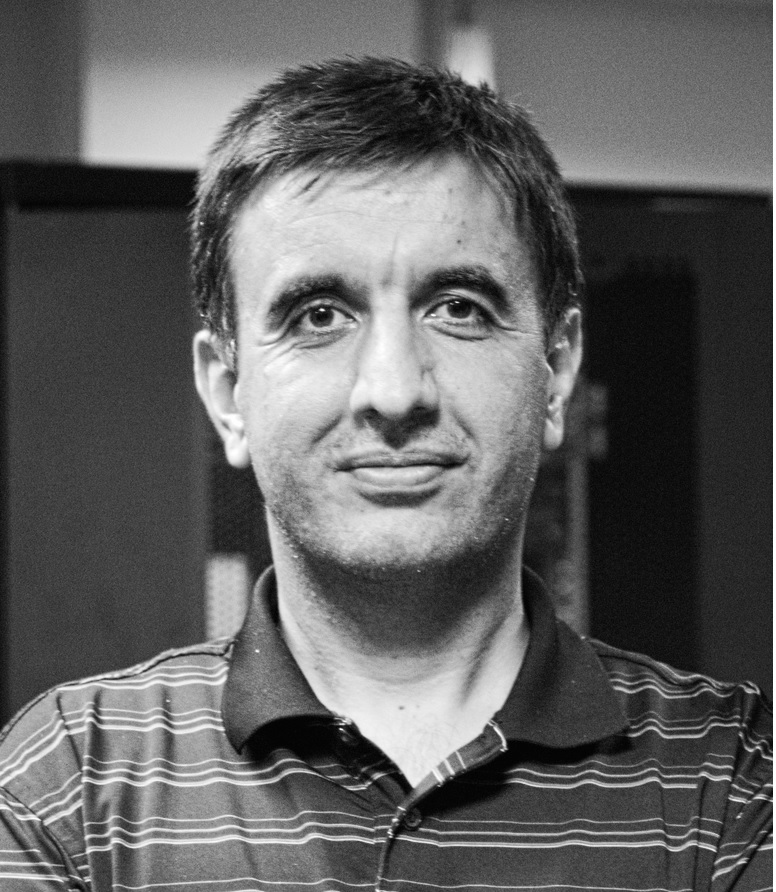}}]{Enver Ozdemir}
	is currently a Professor at Informatics Institute at Istanbul Technical University. He obtained his PhD degree in Mathematics from University of Maryland
	College Park in 2009. He was a member of the Coding Theory and Cryptography Research Group (CCRG), Nanyang Technological University, Singapore from 2010-2014. He is also the deputy director of the National Center for High Performance Computing (UHeM). His research interests include Cryptography, Computational Number Theory and Network Security.
\end{IEEEbiography}

\bibliography{references}
\bibliographystyle{IEEEtran}

\newpage

\vfill

\end{document}